\documentclass[12pt]{article}
\usepackage{color}
\usepackage{xcolor, soul}
\sethlcolor{gray}  
\usepackage{amsmath, amssymb, amsthm}
\usepackage{bm}
\usepackage{braket}
\newtheorem{theorem}{Theorem}
\newtheorem{lemma}{Lemma} 

\newtheorem{proposition}{Proposition} 
\newtheorem{corollary}{Corollary}
\newtheorem{remark}{Remark}
\newtheorem*{proof*}{Proof}
\usepackage{subcaption}
\usepackage{graphicx}
\usepackage{float}
\newcommand{\figwidth}{0.7}

\usepackage{hyperref}
\hypersetup{
setpagesize=false,
 colorlinks=true,%
 linkcolor=magenta,
 citecolor=blue,
}
\begin{document}
\title{{\bf IPS/Zeta Correspondence \\
            for the Domany-Kinzel model}
    \vspace{20mm}}
\author{
    Chusei KIUMI\\
    Graduate School of Engineering Science \\
    Yokohama National University \\
    Hodogaya, Yokohama, 240-8501, JAPAN \\
    e-mail: kiumi-chusei-bf@ynu.jp \\
    \\
    Norio KONNO \\
    Department of Applied Mathematics, Faculty of Engineering \\
    Yokohama National University \\
    Hodogaya, Yokohama, 240-8501, JAPAN \\
    e-mail: konno-norio-bt@ynu.ac.jp \\
    \\
    Yuki OSHIMA\\
    Graduate School of Engineering Science \\
    Yokohama National University \\
    Hodogaya, Yokohama, 240-8501, JAPAN \\
    e-mail: oshima-yuki-ch@ynu.jp
}

\date{\empty}
\maketitle
\clearpage
\begin{abstract}
    Previous studies presented zeta functions by the Konno-Sato theorem or the Fourier analysis for one-particle models, including random walks, correlated random walks, quantum walks, and open quantum random walks. Furthermore, the zeta functions for the multi-particle model with probabilistic or quantum interactions, called the interacting particle system (IPS), were also investigated. In this paper, we focus on the zeta function for a class of IPS, including the Domany-Kinzel model, which is a typical model of the probabilistic IPS in the field of statistical mechanics and mathematical biology.
\end{abstract}

\vspace{10mm}

\begin{small}
\end{small}
\begin{small}
\end{small}
\vspace{10mm}

\section{Introduction}\label{sec01}
In this paper, we study the interacting particle system (IPS), which is known as the multi-particle model, including the probabilistic cellular automata (PCA) and quantum cellular automata (QCA). Our main focus is the IPS-type zeta function for the IPS model, including the Domany-Kinzel model (DK model) as a special case. The DK model was introduced by Domany and Kinzel (1984) \cite{DK} to describe the phase transition phenomenon. It is defined by two parameters $p$ and $q$ $(0\leq p,q\leq 1)$, and it has oriented bond percolation and oriented site percolation as special cases, which are essential models in the field of statistical mechanics and mathematical biology.

Furthermore, a relation between the IPS and zeta function is called IPS/Zeta Correspondence, and it is introduced in Komatsu, Konno, and Sato~\cite{IPS}. The zeta correspondence is studied for various other models~\cite{SWZ,FCTM,FGZ,FVF}. In particular, Walk/Zeta Correspondence~\cite{SWZ} is based on the one-particle system in the lattice space; on the other hand, our focus is on the multi-particle system in the configuration space. This paper also pays attention to the DK model, which is not addressed on~\cite{IPS}.

In Section \ref{sec02}, we give definitions for the multi-particle system, DK model and IPS-zeta function. Section \ref{sec03} is devoted to the analysis of the zeta function on the class of IPS. Moreover, our main results are presented. Finally, Section \ref{sec04} gives a summary.
\section{Definition}
\label{sec02}
\subsection{Multi-particle system}
\label{sec02-1}
In this section, we consider a multi-particle system. First, we define the $N$-site path space by $ \mathbb{P}_{N} =\{0,1,\cdots ,N-1\}\ (N\geq 2)$. For each position $ x\in \mathbb{P}_{N}$, we give a state $\eta (x) \in \{0,1\}$. Here, we can interpret $\eta (x) =1$ to mean that the particle exists at site $x\in \mathbb{P}_{N}$ and $\eta (x) =0$ to mean that the $x$ is vacant. The configuration space is given as $ \{0,1\}^{\mathbb{P}_{N}}$ with $ 2^{N}$ elements and a configuration $ \eta =(\eta (0) ,\eta (1) ,\dotsc ,\eta (N-1)) \in \{0,1\}^{\mathbb{P}_{N}}$ denotes the existence or non-existence of particles in the path space. For instance, $(0,0,1)\in\{0,1\}^{\mathbb{P}_3}$ implies that the particle exists at position $2$ but does not exist at positions $0$ and $1$. We can also express these states with $\ket{0}$ and $\ket{1}$ which form an orthonormal basis of $\mathbb{C}^2$:
\begin{align*}
     & \ket{0}
    =\begin{bmatrix}
         1 \\
         0
     \end{bmatrix},
    \qquad
    \ket{1}
    =\begin{bmatrix}
         0 \\
         1
     \end{bmatrix}.
\end{align*}
Here, $\mathbb{C}$ denotes the set of complex numbers. Then, a configuration $\eta=(\eta(0),\eta(1),\dots , \eta(N-1))$ can be written as $\ket{\eta(0)}\ket{\eta(1)}\cdots\ket{\eta(N-1)}$, which is an abbreviated notation for the tensor product of the states, i.e., $\ket{\eta(0)}\otimes\ket{\eta(1)}\otimes\cdots\otimes\ket{\eta(N-1)}$. For example, a configuration $(0,1,0)\in\{0,1\}^{\mathbb{P}_3}$ is expressed as
\begin{align*}
     & \ket{0}\ket{1}\ket{0}
    =\ket{0}
    \otimes
    \ket{1}
    \otimes
    \ket{0} =\begin{bmatrix}
                 1
                 \\
                 0
             \end{bmatrix}
    \otimes \begin{bmatrix}
                0
                \\
                1
            \end{bmatrix}
    \otimes \begin{bmatrix}
                1 \\
                0
            \end{bmatrix}
    =\begin{bmatrix}
         1 \\
         0
     \end{bmatrix}
    \otimes
    \begin{bmatrix}
        0 \\
        0 \\
        1 \\
        0
    \end{bmatrix} =
    \begin{bmatrix}
        0 \\
        0 \\
        1 \\
        0 \\
        0 \\
        0 \\
        0 \\
        0
    \end{bmatrix} \in \mathbb{C}^{2^{3}} .
\end{align*}

Next, we introduce two types of operators. The first is the \textit{local} operator $ Q^{(l)}$, which is a $4\times 4$ matrix given by
\begin{align*}
    Q^{(l)} &
    =\begin{bmatrix}
         a_{00}^{00} & a_{00}^{01} & a_{00}^{10} & a_{00}^{11} \\
         a_{01}^{00} & a_{01}^{01} & a_{01}^{10} & a_{01}^{11} \\
         a_{10}^{00} & a_{10}^{01} & a_{10}^{10} & a_{10}^{11} \\
         a_{11}^{00} & a_{11}^{01} & a_{11}^{10} & a_{11}^{11}
     \end{bmatrix} ,
\end{align*}
where $ a_{kl}^{ij} \in \mathbb{C} \ (i,j,k,l\in \{0,1\})$. Here, $ a_{kl}^{ij}$ means the \textit{transition weight} from $ (\eta _{n}(x) ,\eta _{n}(x+1)) =(i,j)$ to $ (\eta _{n+1}( x) ,\eta _{n+1}(x+1)) =(k,l)$ for any $ x=0,1,\dotsc ,N-2$. For example, $ a_{11}^{01}$ means the \textit{transition weight} from $ (\eta _{n}( x) ,\eta _{n}(x+1)) =(0,1)$ to $ (\eta _{n+1}( x) ,\eta _{n+1}(x+1)) =(1,1)$. Note that we impose the important condition $ a_{kl}^{ij} =0\ (j\neq l)$ throughout this paper. This condition implies that $Q^{(l)}$ acting on $ (\eta _{n}( x) ,\eta _{n}(x+1))$ does not change the state on the right site, meaning $\eta _{n+1}(x+1)=\eta _{n}(x+1)$. Thus, $ Q^{(l)}$ becomes
\begin{align*}
    Q^{(l)} &
    =\begin{bmatrix}
         a_{00}^{00} & \cdot       & a_{00}^{10} & \cdot       \\
         \cdot       & a_{01}^{01} & \cdot       & a_{01}^{11} \\
         a_{10}^{00} & \cdot       & a_{10}^{10} & \cdot       \\
         \cdot       & a_{11}^{01} & \cdot       & a_{11}^{11}
     \end{bmatrix} .
\end{align*}
Here, $``\cdot "$ represents $0$. We can see from the next discussion about the $global$ operator that this condition also tells us that the state of the site $x$ at time $n+1$ of the model is determined only by the state of the site $x,\ x+1$ at time $n$. In particular, if $a_{kl}^{ij} \in \{0,1\}$, then this model is called the \textit{cellular automaton} (CA).
Also, if $Q^{(l)}$ is a transposed \textit{stochastic matrix}, which means
\begin{align*}
    a_{00}^{00} +a_{10}^{00} & =a_{01}^{01} +a_{11}^{01} =a_{00}^{10} +a_{10}^{10} =a_{01}^{11} +a_{11}^{11} =1,\ \ \ a_{kl}^{ij} \in [ 0,1],
\end{align*} the model is called the \textit{probabilistic cellular automata} (PCA). Furthermore, when $Q^{(l)}$ is unitary, i.e.,
\begin{align*}
    |a_{00}^{00} |^{2} +|a_{10}^{00} |^{2}
     & =|a_{01}^{01} |^{2} +|a_{11}^{01} |^{2} =|a_{00}^{10} |^{2} +|a_{10}^{10} |^{2} =|a_{01}^{11} |^{2} +|a_{11}^{11} |^{2} =1,
    \\
     & a_{00}^{00}\overline{a_{00}^{10}} +a_{10}^{00}\overline{a_{10}^{10}}
    =
    a_{01}^{01}\overline{a_{01}^{11}} +a_{11}^{01}\overline{a_{11}^{11}} =0,
\end{align*}
we call the model the \textit{quantum cellular automata} (QCA). This QCA is introduced by Konno \cite{B32} as a quantum counterpart of the DK model. Next, we consider the second type of operators called the \textit{global} operator $ Q_{N}^{(g)}$, defined by the $ 2^{N} \times 2^{N}$ matrix as
\begin{align*}
    Q_{N}^{(g)} & =\left(I_{2} \otimes I_{2} \otimes \cdots \otimes I_{2} \otimes Q^{(l)}\right)\left(I_{2} \otimes I_{2} \otimes \cdots \otimes Q^{(l)} \otimes I_{2}\right)\nonumber  \\
                & ~~~~\cdots \left(I_{2} \otimes Q^{(l)} \otimes \cdots \otimes I_{2} \otimes I_{2}\right)\left(Q^{(l)} \otimes I_{2} \otimes \cdots \otimes I_{2} \otimes I_{2}\right)
    \nonumber                                                                                                                                                                           \\
                & =\left(I_{2^{(N-2)}} \otimes Q^{(l)}\right)\left(I_{2^{(N-3)}} \otimes Q^{(l)} \otimes I_{2}\right)
    \nonumber
    \\
                & ~~~~~~~~~~~~~~\cdots \left(I_{2} \otimes Q^{(l)} \otimes I_{2^{(N-3)}}\right)\left(Q^{(l)} \otimes I_{2^{(N-2)}}\right) ,
\end{align*}
where $I_{n}$ is the $ n\times n$ identity matrix. In particular, if $ N=2$, $ Q_{N}^{(g)}$ becomes $ Q^{(l)}$. To make the discussion concise, we put $ Q_{1}^{(g)} =I_{2}$. Finally, for the initial state $\eta _{0} \in \mathbb{C}^{2^{N}},$ the evolution of the IPS on $\mathbb{P}_{N}$ is determined by
\begin{align*}
    \eta _{n} & =\left(Q_{N}^{(g)}\right)^{n} \eta _{0} \ \ (n\in \mathbb{Z}_{\geq }),
\end{align*}
where $\mathbb{Z}_{\geq}$ denoets the set of non-negative integers. For example, when $N=3$ and a configuration $\eta=(0,0,1)$, we have
\begin{align*}
    Q_{3}^{(g)}(0,0,1) & =a_{00}^{00} a_{01}^{01} (0,0,1) +a_{00}^{00}a_{11}^{01} (0,1,1) +a_{10}^{00}a_{01}^{01} (1,0,1) +a_{10}^{00}a_{11}^{01} (1,1,1).
\end{align*}
For PCA, the probability that particles exist in a configuration $(0,1,1)$ is $a_{00}^{00}a_{11}^{01}$, and for QCA, it becomes $ \left|a_{00}^{00} a_{11}^{01} \right| ^{2}$.

\subsection{DK model}
This section gives the definition and property of the DK model.
The space-time $ S$ in which the DK model is considered is
\begin{align*}
    S & =\{s=(x,n) \in \mathbb{Z} \times \mathbb{Z}_{\geq}\} ,
\end{align*}
where the integer lattice $\mathbb{Z}$ represents the space and the set of non-negative integers $\mathbb{Z}_{\geq}$ represents time. The DK model is defined by two parameters $ p,q\in [ 0,1]$. We define the evolution as
\begin{align*}
    P\left(x\in \xi _{n+1}^{A} |\xi _{n}^{A}\right) & =f\left(
    \left|\xi _{n}^{A} \cap \{x,x+1\}\right|
    \right),
\end{align*}
where $f$ is defined by $p$ and $q$ as follows:
\begin{align*}
    f(0) =0,\ \  & f(1) =p,\ \ f(2) =q.
\end{align*}
Here, $\xi _{n}^{A}(\subset \mathbb{Z})$ denotes a set of sites that particles exist at time $n$ starting from $A(\subset\mathbb{Z})$, $P(E_1|E_2)$ means the conditional probability of an event $E_1$ given an event $E_2$, and $|B|$ is the number of elements in a set $B$. The visual representation of the dynamics is expressed in Table \ref{table:1}.
Here, the dynamics of the DK model can also be seen as the multi-particle system defined in subsection \ref{sec02-1} with the \textit{local} operator $Q_{DK}^{(l)}$:
\begin{align*}
    Q_{DK}^{(l)} & =\begin{bmatrix}
                        1     & \cdot & 1-p   & \cdot \\
                        \cdot & 1-p   & \cdot & 1-q   \\
                        \cdot & \cdot & p     & \cdot \\
                        \cdot & p     & \cdot & q
                    \end{bmatrix} .
\end{align*}
In particular, if $q=p,$ the probability that each site opens is $p$, thus the DK model becomes the oriented site percolation. Also, if $q=1-(1-p)^{2}$, the probability that each bond opens is $p$, hence the DK model becomes the oriented bond percolation (see \cite{B10} for details).
\begin{table}[h!]
    {\small
        \begin{align*}
             & \begin{array}{ c|c c|c c|c c|c c }
                                      & \ \ \ x\ \ \  & x+1   & \ \ \ x\ \ \     & x+1     & \ \ \ x\ \ \  & x+1   & \ \ \ x\ \ \  & x+1     \\
                   \hline
                   \text{step }n      & \circ         & \circ & \circ            & \bullet & \bullet       & \circ & \bullet       & \bullet \\
                   \text{step }n+1    & \bullet       &       & \ \ \bullet \ \  &         & \bullet       &       & \bullet       &         \\
                   \text{probability} & 0             &       & p                &         & p             &       & \ q\          &
               \end{array}
        \end{align*}
        \begin{align*}
             & \begin{array}{ c|c c|c c|c c|c c }
                                      & \ \ \ x\ \ \  & x+1   & \ \ \ x\ \ \  & x+1     & \ \ \ x\ \ \  & x+1   & \ \ \ x\ \ \  & x+1     \\
                   \hline
                   \text{step }n      & \circ         & \circ & \circ         & \bullet & \bullet       & \circ & \bullet       & \bullet \\
                   \text{step }n+1    & \circ         &       & \circ         &         & \circ         &       & \circ         &         \\
                   \text{probability} & 1             &       & 1-p           &         & 1-p           &       & 1-q\          &
               \end{array}
        \end{align*}}
    \caption{Dynamics of the DK model, We put $ ``\circ "$ denoting that $x$ is vacant, and $ ``\bullet "$ means that a particle exists}
    \label{table:1}
\end{table}
In addition, if $0\leq p\leq q\leq 1$, the DK model is called ``attractive"; the more particles there are, the more likely it is to generate particles. However, if $0\leq q< p\leq 1$, the DK model is called ``non-attractive", and more particles do not necessarily mean that particles are more likely to be produced. We cannot anticipate the global dynamic of particles since it has the effect of both increasing and reducing particles locally.
We also set $\xi _{n}^{0}$ to denote the DK model at time $n$ starting from the origin. Then, the survival probability that particles will continue to exist is given by
\begin{align*}
    \rho (p,q) & =P\left(\xi _{n}^{0} \neq \emptyset \text{ for all } n\geq 0\right) .
\end{align*}
The sequence $ \left\{\xi _{n}^{0} \neq \emptyset \right\}$ is a monotonically non-increasing function, and the strict definition of $ \rho (p,q)$ is
\begin{align*}
    \rho (p,q) & =\lim _{n\rightarrow \infty } P\left(\xi _{n}^{0} \neq \emptyset \right) .
\end{align*}
Also, the survival probability of the DK model which started from $A$ is given as
\begin{align*}
    \sigma (A) & =\sigma (A:p,q) =P\left(\xi _{n}^{A} \neq \emptyset\text{ for all }n\geq 0\right) .
\end{align*}
By contrast, the extinction probability is defined by
\begin{align*}
    \nu (A) =\nu (A:p,q) & =1-\sigma (A).
\end{align*}
Furthermore, when we fix $ q\in [ 0,1]$, we introduce two critical probabilities.
\begin{align*}
    p_{c}(q)     & =\sup \left\{p\in [ 0,1] :\rho \left(p^{\prime } ,q\right) =0\text{ for all }p^{\prime } \in [ 0,p]\right\} , \\
    p_{c}^{*}(q) & =\inf\left\{p\in [ 0,1] :\rho \left(p^{\prime } ,q\right)  >0\text{ for all }p^{\prime } \in [ p,1]\right\} .
\end{align*}
We can easily check that
\begin{align*}
    0\leq p_{c}(q) \leq p_{c}^{*}(q) \leq 1.
\end{align*}
These critical probabilities are not mathematically derived; however, the Monte Carlo Simulation conducted in \cite{lubeck} conjectured that $p_c(q)=p^*_c(q)$ for any $q\in[0,1]$, and $\rho(p,q_0)$ is monotonically non-decreasing function with respect to $p$ for any fixed $q_0\in[0,1]$. The critical line $\{(p_{c}(q) ,q) :q\in [ 0,1]\}$ given by the Monte Carlo Simulation implies that the area is divided into two areas; the survival region $D_{s}$ and the extinction region $D_{e}$ defined by
\begin{align*}
    D_{s} =\{(p,q) \in D:\rho (p,q)  >0\} ,\ \  & \ \ D_{e} =\{(p,q) \in D:\rho (p,q) =0\} .
\end{align*}

For a simple case, $q=1$, it is shown that $ p_{1} \leq p_{2}$ implies $ \rho (p_{1} ,1) \leq \rho (p_{2} ,1)$ by the coupling method \cite{Schinazi1999}, and $\rho(p,1)$ is solved for all $p$ as follows:
\begin{align*}
    \rho (p,1) & =\begin{cases}
                      1-(1-p)^{2} /p^{2}  & (1/2\leq p\leq 1) , \\
                      \ \ \ \ \ \ \ \ \ 0 & (0\leq p\leq 1/2) .
                  \end{cases}
\end{align*}
However, the critical line is not derived mathematically for the general case, and even the monotonicity of $\rho$ is not proved; the DK model has a lot of unsolved problems.


\subsection{The IPS-type zeta function}
In this subsection, we consider the IPS-type zeta function introduced in \cite{IPS}, defined by
\begin{align}
    \overline{\zeta }\left(Q^{(l)} ,\mathbb{P}_{N} ,u\right) & =\det\left(I_{2^{N}} -uQ_{N}^{(g)}\right)^{-1/2^{N}} \ \ \ \ (N\in \mathbb{Z}_{ >}) ,
    \label{zetadet}
\end{align}
where $\mathbb{Z}_{ >}$ denotes the set of positive integers. For $N=1$, we have
\begin{align}
    \overline{\zeta }\left(Q^{(l)} ,\mathbb{P}_{1} ,u\right) & =\det(I_{2} -uI_{2})^{-1/2} =(1-u)^{-1} .
    \label{zeta}
\end{align}
The IPS-type zeta function is based on the multi-particle system on the configuration space. We define $C_{r}\left(Q^{(l)} ,\mathbb{P}_{N}\right)$ as follows:
\begin{align*}
    \overline{\zeta }\left(Q^{(l)} ,\mathbb{P}_{N} ,u\right) & =\exp\left(\sum _{r=1}^{\infty }\frac{C_{r}\left(Q^{(l)} ,\mathbb{P}_{N}\right)}{r} u^{r}\right) \ \ \ \ \ \ \ (N\in \mathbb{Z}_{ >}).
\end{align*}
Then $C_{r}\left(Q^{(l)} ,\mathbb{P}_{N}\right)$ means a rate of configurations that return to the initial configuration at $r$ step. The visual representation of the case $(p,q)=(1,0)$ is shown in Table \ref{table:2}.
\begin{table}[H]
    \centering
    {\small
        \begin{align*}
             & \begin{array}{ c|c|c|c|c|c|c|c|c }
                   r=0
                       & 0\ 0\ 0
                       & 0\ 0\ 1
                       & 0\ 1\ 0
                       & 0\ 1\ 1
                       & 1\ 0\ 0
                       & 1\ 0\ 1
                       & 1\ 1\ 0
                       & 1\ 1\ 1
                   \\
                   \hline
                   r=1 &
                   \colorbox{lightgray}{0\ 0\ 0}
                       & 0\ 1\ 1
                       & 1\ 1\ 0
                       & 1\ 0\ 1
                       & \colorbox{lightgray}{1\ 0\ 0}
                       & 1\ 1\ 1
                       & 0\ 1\ 0
                       & 0\ 0\ 1
                   \\
                   r=2
                       & \colorbox{lightgray}{0\ 0\ 0}
                       & 1\ 0\ 1
                       & \colorbox{lightgray}{0\ 1\ 0}
                       & 1\ 1\ 1
                       & \colorbox{lightgray}{1\ 0\ 0}
                       & 0\ 0\ 1
                       & \colorbox{lightgray}{1\ 1\ 0}
                       & 0\ 1\ 1
                   \\
                   r=3
                       & \colorbox{lightgray}{0\ 0\ 0}
                       & 1\ 1\ 1
                       & 1\ 1\ 0
                       & 0\ 0\ 1
                       & \colorbox{lightgray}{1\ 0\ 0}
                       & 0\ 1\ 1
                       & 0\ 1\ 0
                       & 1\ 0\ 1
                   \\
                   r=4
                       & \colorbox{lightgray}{0\ 0\ 0}
                       & \colorbox{lightgray}{0\ 0\ 1}
                       & \colorbox{lightgray}{0\ 1\ 0}
                       & \colorbox{lightgray}{0\ 1\ 1}
                       & \colorbox{lightgray}{1\ 0\ 0}
                       & \colorbox{lightgray}{1\ 0\ 1}
                       & \colorbox{lightgray}{1\ 1\ 0}
                       & \colorbox{lightgray}{1\ 1\ 1}
               \end{array}
        \end{align*}
    }
    \caption{Dynamics of the case $N=3$, $(p,q) =(1,0)$. Configurations which match the initial configurations are highlighted. Here, $C_{1}=\frac{1}{4},C_{2}=\frac{1}{2},C_3=\frac{1}{4},C_4=1$, and they correspond to the rate of highlighted configurations. Also, since $C_4=1$, i.e., the configuration of $r=4$ is identical to the initial configuration, thus we can say that the period of this dynamic is $4$.}
    \label{table:2}
\end{table}
Let $ \lambda _{j} \ \left(j=1,2,\dotsc ,2^{N}\right)$ be eigenvalues of $ Q_{N}^{(g)}$ and $ \mathrm{tr}(A)$ be the trace of a square matrix $A$. The previous study \cite{IPS} gives the following theorem.
\begin{theorem}
    \begin{align}
        \label{eq:log}
        \log\left\{\overline{\zeta }\left(Q^{(l)} ,\mathbb{P}_{N} ,u\right)\right\} & =\sum _{r=1}^{\infty }\left\{\frac{1}{2^{N}}\mathrm{tr}\left(\left(Q_{N}^{(g)}\right)^{r}\right)\right\}\frac{u^{r}}{r} .
    \end{align}
\end{theorem}
Combining \eqref{zetadet} and \eqref{eq:log} implies the following important result.
\begin{align*}
    C_{r}\left(Q^{(l)} ,\mathbb{P}_{N}\right) & =\frac{1}{2^{N}}\mathrm{tr}\left(\left(Q_{N}^{(g)}\right)^{r}\right) \ \ \ \ \ (r,N\in \mathbb{Z}_{ >}) .
\end{align*}
For example, we consider the \textit{global} operator $ Q_{3}^{(g)}$ of the DK model. By a direct calculation, we obtain the eigenvalues of $ Q_{3}^{(g)}$:
\begin{align*}
    \mathrm{Spec}\left(Q_{3}^{(g)}\right) & =\left\{[ 1]^{2} ,[ p]^{2} ,[ q-p]^{1} ,[ p(q-p)]^{1} ,[ \lambda _{+}]^{1} ,[ \lambda _{-}]^{1}\right\},
\end{align*}
where
\begin{align*}
    \lambda _{\pm } & =\frac{1}{2}\left(-k\pm \sqrt{k^{2} -4p(p-q)^{2}}\right),\ k =p^{2}-q^{2} +pq-p .
\end{align*}
In addition, we get
\begin{align*}
    C_{r}=\frac{1}{2^{3}}\left\{2+2p^{r} +(q-p)^{r} +(p(q-p))^{r} +(\lambda _{+})^{r} +(\lambda _{-})^{r}\right\} .
\end{align*}
Finally, the zeta function becomes
\begin{align*}
     & \overline{\zeta }\left(Q_{DK}^{(l)} ,\mathbb{P}_{3} ,u\right)
    \\
     & =\left[(1-u)^{2}(1-pu)^{2}\{1-(q-p) u\}\{1-p(q-p) u\} \ (1-\lambda _{+} u)(1-\lambda _{-} u)\right]^{-\frac{1}{2^{3}}}.
\end{align*}

\section{Results \label{sec03}}
In the previous section, we showed that $ \mathrm{tr}\left(\left(Q_{N}^{(g)}\right)^{r}\right)$ is essential for the calculations of the zeta function, hence we focus on the analysis on the trace of the $global$ operator in this subsection.
We decompose $Q_{N}^{(g)}$ into $2^{N-1} \times 2^{N-1}$ block matrices, $E_{N} ,F_{N} ,G_{N}$ and $H_{N}$, as below:
\[
    \ Q_{N}^{(g)} =\left[\begin{array}{ l l }
            E_{N} & F_{N} \\
            G_{N} & H_{N}
        \end{array}\right],\quad (N\geq1).
\]
Then, we get the important relation.
\begin{lemma}
    \label{lem}
    $Q_{N+1}^{(g)}$ can be expressed as
    \[
        \ Q_{N+1}^{(g)} =\left[\begin{array}{ c c c c }
                a_{00}^{00} E_{N} & a_{01}^{01} F_{N} & a_{00}^{10} E_{N} & a_{01}^{11} F_{N} \\
                a_{00}^{00} G_{N} & a_{01}^{01} H_{N} & a_{00}^{10} G_{N} & a_{01}^{11} H_{N} \\
                a_{10}^{00} E_{N} & a_{11}^{01} F_{N} & a_{10}^{10} E_{N} & a_{11}^{11} F_{N} \\
                a_{10}^{00} G_{N} & a_{11}^{01} H_{N} & a_{10}^{10} G_{N} & a_{11}^{11} H_{N}
            \end{array}\right] .
    \]
\end{lemma}
\begin{proof}
    From the definition, we have the following formula:
    \[
        Q_{N+1}^{(g)} =\left(I_{2} \otimes Q_{N}^{(g)}\right)\left(Q^{(l)} \otimes I_{2^{(N-1)}}\right) \ (N\geq 1) .
    \]
    By a direct calculation, we get
    \[
        \left(I_{2} \otimes Q_{N}^{(g)}\right) =\left[\begin{array}{ c c c c }
                E_{N}         & F_{N}         & O_{2^{(N-1)}} & O_{2^{(N-1)}} \\
                G_{N}         & H_{N}         & O_{2^{(N-1)}} & O_{2^{(N-1)}} \\
                O_{2^{(N-1)}} & O_{2^{(N-1)}} & E_{N}         & F_{N}         \\
                O_{2^{(N-1)}} & O_{2^{(N-1)}} & G_{N}         & H_{N}
            \end{array}\right]
    \]
    and
    \[
        \left(Q^{(l)} \otimes I_{2^{(N-1)}}\right) =\left[\begin{array}{ c c c c }
                a_{00}^{00} I_{2^{(N-1)}} & O_{2^{(N-1)}}             & a_{00}^{10} I_{2^{(N-1)}} & O_{2^{(N-1)}}             \\
                O_{2^{(N-1)}}             & a_{01}^{01} I_{2^{(N-1)}} & O_{2^{(N-1)}}             & a_{01}^{11} I_{2^{(N-1)}} \\
                a_{10}^{00} I_{2^{(N-1)}} & O_{2^{(N-1)}}             & a_{10}^{10} I_{2^{(N-1)}} & O_{2^{(N-1)}}             \\
                O_{2^{(N-1)}}             & a_{11}^{01} I_{2^{(N-1)}} & O_{2^{(N-1)}}             & a_{11}^{11} I_{2^{(N-1)}}
            \end{array}\right] .
    \]
    Here, $O_n$ denotes the $n\times n$ zero matrix. Hence, $Q_{N+1}^{(g)}$ becomes
    \[
        \begin{aligned}
            Q_{N+1}^{(g)} & =\left(I_{2} \otimes Q_{N}^{(g)}\right)\left(Q^{(l)} \otimes I_{2^{(N-1)}}\right)    \\
                          & =\left[\begin{array}{ c c c c }
                                           a_{00}^{00} E_{N} & a_{01}^{01} F_{N} & a_{00}^{10} E_{N} & a_{01}^{11} F_{N} \\
                                           a_{00}^{00} G_{N} & a_{01}^{01} H_{N} & a_{00}^{10} G_{N} & a_{01}^{11} H_{N} \\
                                           a_{10}^{00} E_{N} & a_{11}^{01} F_{N} & a_{10}^{10} E_{N} & a_{11}^{11} F_{N} \\
                                           a_{10}^{00} G_{N} & a_{11}^{01} H_{N} & a_{10}^{10} G_{N} & a_{11}^{11} H_{N}
                                       \end{array}\right] .
        \end{aligned}
    \]
\end{proof}
Hence, the next result is immediately shown.
\begin{corollary}
    \label{cor}
    \begin{align}
        \label{lemma:rel0}
         & E_{N+1} +G_{N+1} =F_{N+1} +H_{N+1} =Q_{N}^{(g)}
    \end{align}
    and following relations hold:
    \begin{align}
        \label{lemma:rel1}
         & \mathrm{tr}(E_{N}) =a_{00}^{00}\mathrm{tr}(E_{N-1}) +a_{01}^{01}\mathrm{tr}(H_{N-1}), \\
        \label{lemma:rel2}
         & \mathrm{tr}(H_{N}) =a_{10}^{10}\mathrm{tr}(E_{N-1}) +a_{11}^{11}\mathrm{tr}(H_{N-1}).
    \end{align}
\end{corollary}
By using these relations, (\ref{lemma:rel1}) and (\ref{lemma:rel2}), we can calculate $\operatorname{tr}\left(Q_{N}^{(g)}\right)$ as the following propositions.
\begin{proposition}
    \label{prop1}
    For $N>1$,
    \begin{align*}
        \ \mathrm{tr}\left(Q_{N}^{(g)}\right) & =\sum _{i_{1}}\sum _{i_{2}} \cdots \sum _{i_{N}} a_{i_{1} i_{2}}^{i_{1} i_{2}} a_{i_{2} i_{3}}^{i_{2} i_{3}} \cdots a_{i_{N-1} i_{N}}^{i_{N-1} i_{N}} .\
    \end{align*}
\end{proposition}
\begin{proof}
    Since $\mathrm{tr}\left(Q_{N}^{(g)}\right) =\operatorname{tr}(E_{N}) +\operatorname{tr}(H_{N})$ is satisfied, the proof is complete if the next statement is shown for $N>1$:
    \begin{align*}
         & \operatorname{tr}(E_{N}) =\sum _{i_{1}} \cdots \sum _{i_{N-1}} a_{0i_{1}}^{0i_{1}} a_{i_{1} i_{2}}^{i_{1} i_{2}} \cdots a_{i_{N-2} i_{N-1}}^{i_{N-2} i_{N-1}}, \\
         & \operatorname{tr}(H_{N}) =\sum _{i_{1}} \cdots \sum _{i_{N-1}} a_{1i_{1}}^{1i_{1}} a_{i_{1} i_{2}}^{i_{1} i_{2}} \cdots a_{i_{N-2} i_{N-1}}^{i_{N-2} i_{N-1}}.
    \end{align*}
    If $N=2$, the above statement holds since
    \begin{align*}
         & \operatorname{tr}(E_{2}) =\sum _{i_{1}} a_{0i_{1}}^{0i_{1}} ,\ \operatorname{tr}(H_{2}) =\sum _{i_{1}} a_{1i_{1}}^{1i_{1}} .
    \end{align*}
    Next, we assume that the statement holds for $N-1$, then relations (\ref{lemma:rel1}) and (\ref{lemma:rel2}) in Corollary \ref{cor} gives
    \begin{align*}
        \mathrm{tr}(E_{N}) & =a_{00}^{00}\ \mathrm{tr}(E_{N-1}) +a_{01}^{01}\ \mathrm{tr}(H_{N-1})                                                                 \\
                           & =\sum _{i_{1}} \cdots \sum _{i_{N-1}} a_{0i_{1}}^{0i_{1}} a_{i_{1} i_{2}}^{i_{1} i_{2}} \cdots a_{i_{N-2} i_{N-1}}^{i_{N-2} i_{N-1}}, \\
        \mathrm{tr}(H_{N}) & =a_{10}^{10}\ \mathrm{tr}(E_{N-1}) +a_{11}^{11}\ \mathrm{tr}(H_{N-1})                                                                 \\
                           & =\sum _{i_{1}} \cdots \sum _{i_{N-1}} a_{1i_{1}}^{1i_{1}} a_{i_{1} i_{2}}^{i_{1} i_{2}} \cdots a_{i_{N-2} i_{N-1}}^{i_{N-2} i_{N-1}}.
    \end{align*}
    Therefore, we get the desired conclusion.
\end{proof}
We can also express a trace of the $global$ operator as the following result.
\begin{proposition}
    \label{prop2}
    \begin{align*}
         & \mathrm{tr}\left(Q_{N}^{(g)}\right) =\frac{x_{+}^{n-1} \Lambda _{+} -x_{-}^{n-1} \Lambda _{-}}{a_{01}^{01}(x_{+} -x_{-})} ,
    \end{align*}
    where
    \[
        \Lambda _{\pm } =\left(a_{01}^{01} -a_{00}^{00} +x_{\pm }\right)\left(a_{00}^{00} +a_{01}^{01} -x_{\mp }\right)
    \]
    and
    \[
        x_{\pm } =\frac{\left(a_{00}^{00} +a_{11}^{11}\right) \pm \sqrt{\left(a_{00}^{00} -a_{11}^{11}\right)^{2} +4a_{01}^{01} a_{10}^{10}}}{2} .
    \]
\end{proposition}
\begin{proof}
    By substituting (\ref{lemma:rel1}) to (\ref{lemma:rel2}), we obtain
    \[
        \mathrm{tr}(E_{N+1}) =\left(a_{00}^{00} +a_{11}^{11}\right)\mathrm{tr}(E_{N}) +\left(a_{01}^{01} a_{10}^{10} -a_{00}^{00} a_{11}^{11}\right)\mathrm{tr}(E_{N-1}).
    \]
    We set
    \[
        x_{\pm } =\frac{\left(a_{00}^{00} +a_{11}^{11}\right) \pm \sqrt{\left(a_{00}^{00} -a_{11}^{11}\right)^{2} +4a_{01}^{01} a_{10}^{10}}}{2},
    \]
    which are two solutions of the following characteristic equation:
    \[
        x^{2} -\left(a_{00}^{00} +a_{11}^{11}\right) x-\left(a_{01}^{01} a_{10}^{10} -a_{00}^{00} a_{11}^{11}\right) =0.
    \]
    Hence, we can write
    \begin{align*}
         & \mathrm{tr}(E_{N+1}) -x_{+}\mathrm{tr}(E_{N}) =x_{-}^{N-1}(\mathrm{tr}(E_{2}) -x_{+}\mathrm{tr}(E_{1})), \\
         & \mathrm{tr}(E_{N+1}) -x_{-}\mathrm{tr}(E_{N}) =x_{+}^{N-1}(\mathrm{tr}(E_{2}) -x_{-}\mathrm{tr}(E_{1})).
    \end{align*}
    Note that,
    \[
        \mathrm{tr}(E_{1}) =1,\ \mathrm{tr}(E_{2}) =a_{00}^{00} +a_{01}^{01}.
    \]
    Thus, we can calculate the trace of $E_{N}$ as
    \[
        \mathrm{tr}(E_{N}) =\frac{x_{+}^{N-1}\left(a_{00}^{00} +a_{01}^{01} -x_{-}\right) -x_{-}^{N-1}\left(a_{00}^{00} +a_{01}^{01} -x_{+}\right)}{x_{+} -x_{-}} .
    \]
    Also, from (\ref{lemma:rel1}), we have
    \[
        \begin{aligned}
            \mathrm{tr}\left(Q_{N}^{(g)}\right) & =\mathrm{tr}(E_{N}) +\mathrm{tr}(H_{N})                                                              \\
                                                & =\frac{\mathrm{tr}(E_{N+1}) +\left(a_{01}^{01} -a_{00}^{00}\right)\mathrm{tr}(E_{N})}{a_{01}^{01}} .
        \end{aligned}
    \]
    Therefore, the statemenit is proved by calculating this formula directly.
\end{proof}
Next, we focus on the eigenvalues of the $global$ operator. We write the set of eigenvalues of the $n\times n$ matrix $A$ as
\[
    \operatorname{Spec}(A) =\left\{[ \lambda _{1}]^{m_{1}} ,[ \lambda _{2}]^{m_{2}} ,\dotsc ,[ \lambda _{k}]^{m_{k}}\right\},
\]
where $m_{j} \in \mathbb{Z}_{ >}$ $(1\leq j\leq k)$ denotes the multiplicity of the eigenvalue $\lambda _{j}$. Here, $m_1+m_2+\cdots +m_k=n$. Also, we define the union $\cup$ as the addition of the multiplicities of each eigenvalue in the left and right. For example, we assume $\operatorname{Spec}(A)$ and $\operatorname{Spec}(B)$ have common eigenvalues $k_{1} ,k_{2}$ and expressed as
\begin{align*}
     & \operatorname{Spec}(A) =\left\{[ k_{1}]^{m_{1}} ,[ k_{2}]^{m_{2}} ,[ a_{1}]^{m_{3}} ,[ a_{2}]^{m_{4}}\right\} ,             \\
     & \operatorname{Spec}(B) =\left\{[ k_{1}]^{m_{1}^{\prime }} ,[ k_{2}]^{m_{2}^{\prime }} ,[ b_{1}]^{m_{3}^{\prime }}\right\} .
\end{align*}
Then,
\[
    \operatorname{Spec}(A) \cup \ \operatorname{Spec}(B) =\left\{[ k_{1}]^{m_{1} +m_{1}^{\prime }} ,[ k_{2}]^{m_{1} +m_{1}^{\prime }} ,[ a_{1}]^{m_{3}} ,[ a_{2}]^{m_{4}} ,[ b_{1}]^{m_{3}^{\prime }}\right\} .
\]
By using the above notations, we get the following theorem.
\begin{theorem}
    \label{theo2}
    \begin{align*}
         & \operatorname{Spec}\left(Q_{N+1}^{(g)}\right)    \\
         & \ \ =\operatorname{Spec}\left(Q_{N}^{(g)}\right)
        \cup
        \operatorname{Spec}
        \left(Q_{N}^{(g)}
        \left[
                \begin{array}{ c c }
                    \left(a_{10}^{10} -a_{10}^{00}\right) I_{2^{N-1}} & O_{2^{N-1}}
                    \\
                    O_{2^{N-1}}                                       & \left(a_{11}^{11} -a_{11}^{01}\right) I_{2^{N-1}}
                \end{array}
                \right]
        \right).
    \end{align*}
    In particular, when $t=a_{10}^{10} -a_{10}^{00} =a_{11}^{11} -a_{11}^{01}$, then, the spectrum of $Q_{N}^{(g)}$ can be explicitly expressed as follows:
    \begin{align*}
        \operatorname{Spec}\left(Q_{N}^{(g)}\right) & =\operatorname{Spec}\left(Q_{N-1}^{(g)}\right)\cup\left(  t\times \operatorname{Spec}\left(Q_{N-1}^{(g)}\right)\right) \\
                                                    & =\bigcup _{k=0}^{N-1}\left[ t^{k}\right]^{2\binom{N-1}{k}}.
    \end{align*}
\end{theorem}
\begin{proof}
    Here, we use the important property (\ref{lemma:rel0}), and we can proceed to the following conversion
    \begin{align*}
             & \det\left(\left[\begin{array}{ c c }
                                       Q_{N}^{(g)} -\lambda I_{2^{N}} & Q_{N}^{(g)} -\lambda I_{2^{N}} \\
                                       G_{N+1}                        & H_{N+1} -\lambda I_{2^{N}}
                                   \end{array}\right]\right) =0      \\
        \iff & \det\left(\left[\begin{array}{ c c }
                                       Q_{N}^{(g)} -\lambda I_{2^{N}} & O_{2^{N}}                           \\
                                       G_{N+1}                        & H_{N+1} -G_{N+1} -\lambda I_{2^{N}}
                                   \end{array}\right]\right) =0.
    \end{align*}
    Here, ``$\iff$" denotes ``if and only if". Thus we can decompose the formula as
    \begin{align*}
             & \det\left(\lambda I_{2^{N}} -Q_{N}^{(g)}\right)\det(H_{N+1} -G_{N+1} -\lambda I_{2^{N}}) =0 \\
        \iff & \det\left(\lambda I_{2^{N}} -Q_{N}^{(g)}\right)\det\left(\lambda I_{2^{N}} -
        \left[\begin{array}{ c c }
                      \left(a_{10}^{10} -a_{10}^{00}\right) E_{N} & \left(a_{11}^{11} -a_{11}^{01}\right) F_{N} \\
                      \left(a_{10}^{10} -a_{10}^{00}\right) G_{N} & \left(a_{11}^{11} -a_{11}^{01}\right) H_{N}
                  \end{array}\right]\right) =0.
    \end{align*}
    Here,
    \begin{align*}
         & \left[\begin{array}{ c c }
                         \left(a_{10}^{10} -a_{10}^{00}\right) E_{N} & \left(a_{11}^{11} -a_{11}^{01}\right) F_{N} \\
                         \left(a_{10}^{10} -a_{10}^{00}\right) G_{N} & \left(a_{11}^{11} -a_{11}^{01}\right) H_{N}
                     \end{array}\right]                                                                         \\
         & =Q_{N}^{(g)}\left[\begin{array}{ c c }
                                     \left(a_{10}^{10} -a_{10}^{00}\right) I_{2^{N-1}} & O_{2^{N-1}}                                       \\
                                     O_{2^{N-1}}                                       & \left(a_{11}^{11} -a_{11}^{01}\right) I_{2^{N-1}}
                                 \end{array}\right].
    \end{align*}
    Then, we get the desired result.
\end{proof}
Next, we show the numerical results of the eigenvalues for the DK model. Figures \ref{fig:2}, \ref{fig:3}, and \ref{fig:4} show the histogram of the eigenvalues of $Q^{(g)}_N$ for oriented site percolation, oriented bond percolation, and $q=0$. The histograms represent the count of eigenvalues within each bin, where the bin size is $0.05\times0.05$ from $-1.00$ to $1.00$ for the real and imaginary parts. In Figures \ref{fig:2} and \ref{fig:3}, the eigenvalues are distributed mostly near the real axis, whereas in Figure \ref{fig:4}, they appear to spread equally in both directions of the real and imaginary parts, when $p$ is larger than $1/2$. In particular, when $p=1,\ q=0$, the model is essentially the same as Wolfram's Rule 90 \cite{wolfram2002new} and $Q^{(g)}_N$ becomes unitary, thus all the eigenvalues are on the unit circle.
\begin{figure}
    \begin{subfigure}[H]{0.5\textwidth}
        \centering
        \includegraphics[width=\figwidth\linewidth]{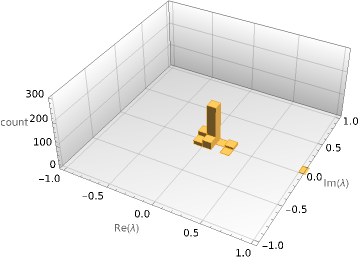}
        \caption{$p=\frac{1}{4},q=\frac{1}{4}$}
    \end{subfigure}
    \begin{subfigure}{0.5\textwidth}
        \centering
        \includegraphics[width=\figwidth\linewidth]{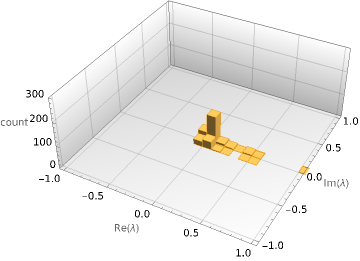}
        \caption{$p=\frac{1}{2},q=\frac{1}{2}$}
    \end{subfigure}
    \begin{subfigure}{0.5\textwidth}
        \centering
        \includegraphics[width=\figwidth\linewidth]{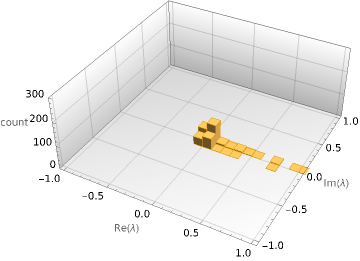}
        \caption{$p=\frac{3}{4},q=\frac{3}{4}$}
    \end{subfigure}
    \begin{subfigure}{0.5\textwidth}
        \centering
        \includegraphics[width=\figwidth\linewidth]{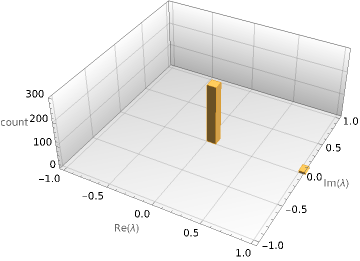}
        \caption{$p=1,q=1$}
    \end{subfigure}
    \caption{Figures show the histograms of the eigenvalues of $Q^{(g)}_N$ in the complex plain, for oriented site percolation $q=p$ with $N=8$.}
    \label{fig:2}
\end{figure}
\begin{figure}
    \begin{subfigure}[H]{0.5\textwidth}
        \centering
        \includegraphics[width=\figwidth\linewidth]{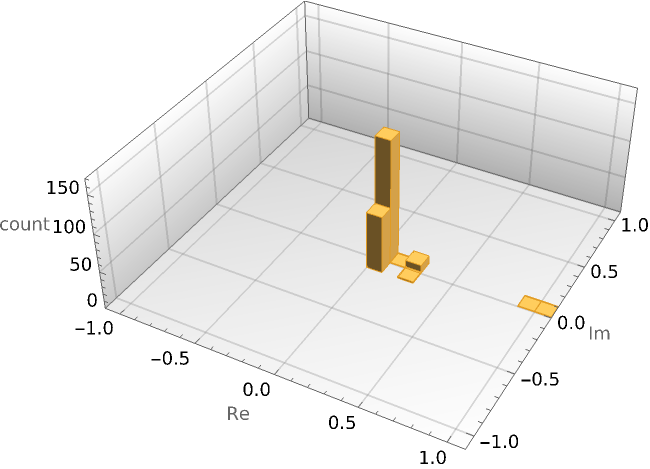}
        \caption{$p=\frac{1}{4},q=\frac{7}{16}$}
    \end{subfigure}
    \begin{subfigure}{0.5\textwidth}
        \centering
        \includegraphics[width=\figwidth\linewidth]{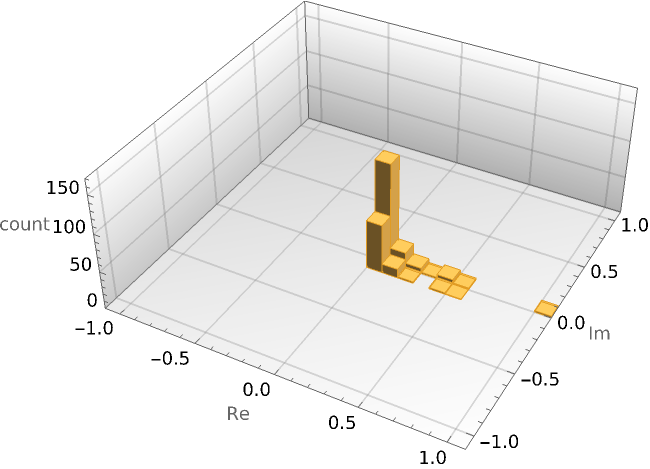}
        \caption{$p=\frac{1}{2},q=\frac{3}{4}$}
    \end{subfigure}
    \begin{subfigure}{0.5\textwidth}
        \centering
        \includegraphics[width=\figwidth\linewidth]{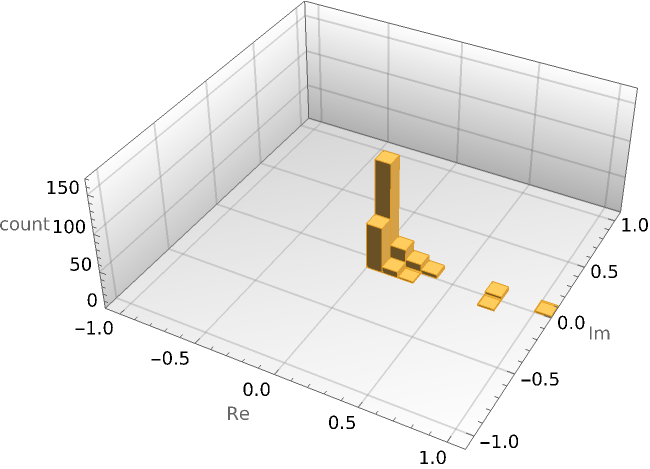}
        \caption{$p=\frac{3}{4},q=\frac{15}{16}$}
    \end{subfigure}
    \caption{Figures show the histograms of the eigenvalues of $Q^{(g)}_N$ in the complex plain, for oriented bond percolation $q=1-(1-p)^{2}$ with $N=8$.}
    \label{fig:3}
\end{figure}
\begin{figure}
    \begin{subfigure}[H]{0.5\textwidth}
        \centering
        \includegraphics[width=\figwidth\linewidth]{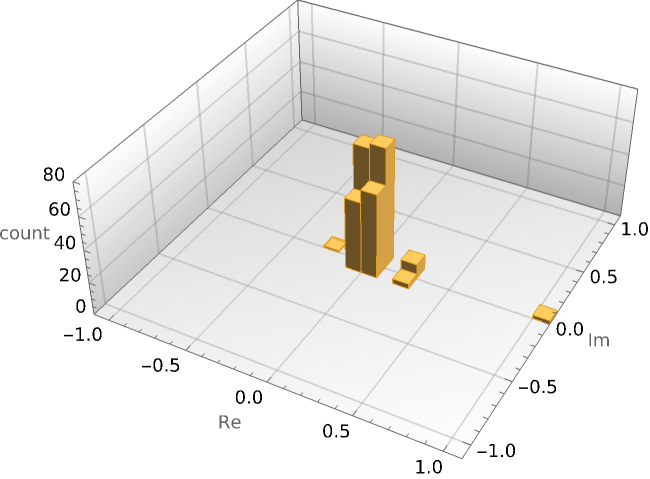}
        \caption{$p=\frac{1}{4},q=0$}
    \end{subfigure}
    \begin{subfigure}{0.5\textwidth}
        \centering
        \includegraphics[width=\figwidth\linewidth]{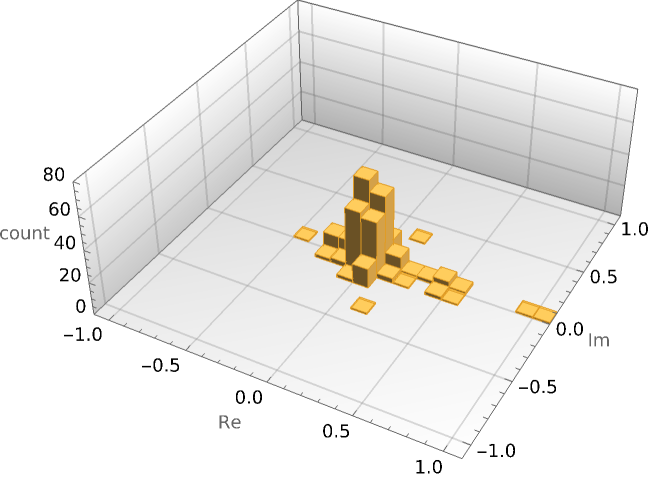}
        \caption{$p=\frac{1}{2},q=0$}
    \end{subfigure}
    \begin{subfigure}{0.5\textwidth}
        \centering
        \includegraphics[width=\figwidth\linewidth]{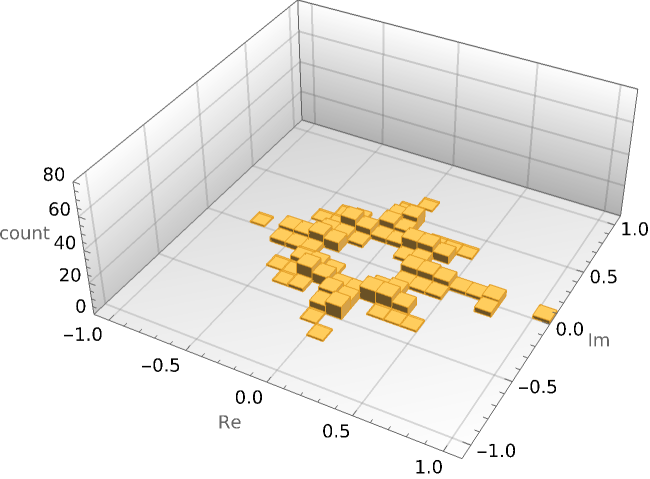}
        \caption{$p=\frac{3}{4},q=0$}
    \end{subfigure}
    \begin{subfigure}{0.5\textwidth}
        \centering
        \includegraphics[width=0.75\linewidth]{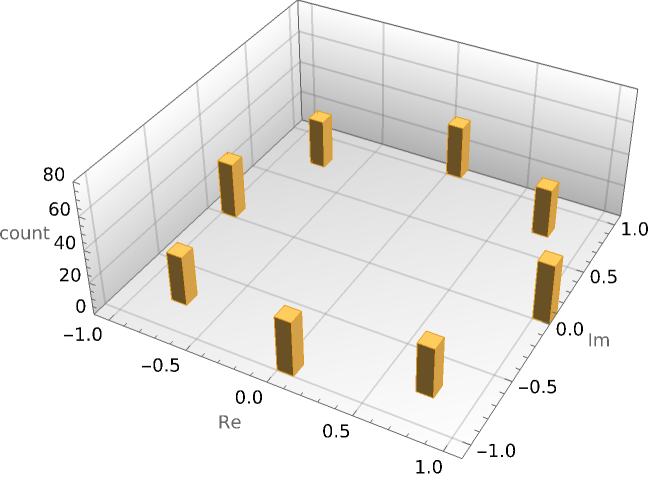}
        \caption{$p=1,q=0$}
    \end{subfigure}
    \caption{Figures show the histograms of the eigenvalues of $Q^{(g)}_N$ in the complex plain for $q=0$ with $N=8$.}
    \label{fig:4}
\end{figure}
\begin{theorem}
    \label{theo3}
    If $t=a_{10}^{10} -a_{10}^{00} =a_{11}^{11} -a_{11}^{01}$,
    \[
        C_{r}\left( Q^{(l)} ,\mathbb{P}_{N}\right) =\left(\frac{1+t^{r}}{2}\right)^{N-1}
    \]
    and
    \begin{align*}
        \log\left\{\overline{\zeta }\left( Q^{(l)} ,\mathbb{P}_{N} ,u\right)\right\} & =-E\left[\log\left( 1-t^{X_{N-1}} u\right)\right] ,
    \end{align*}
    where $X_{n}$ follows the symmetric binomial distribution:
    \[
        P( X_{n} =k) =\binom{n}{k}\left(\frac{1}{2}\right)^{k}\left(\frac{1}{2}\right)^{n-k} .
    \]
\end{theorem}
\begin{proof}
    From Theorem \ref{theo2}, the trace of $\left( Q_{N}^{(g)}\right)^{r}$ can be calculated as follows:
    \begin{align*}
        \operatorname{tr}\left(\left( Q_{N}^{(g)}\right)^{r}\right) & =2\sum _{k=0}^{N-1}\binom{N-1}{k} t^{rk} =2\left( 1+t^{r}\right)^{N-1} .
    \end{align*}
    Thus, we get
    \begin{equation*}
        C_{r}\left( Q^{(l)} ,\mathbb{P}_{N}\right) =\left(\frac{1+t^{r}}{2}\right)^{N-1} .
    \end{equation*}
    Also, we have
    \begin{align*}
        \log\left\{\overline{\zeta }\left( Q^{(l)} ,\mathbb{P}_{N} ,u\right)\right\}
         & =\sum _{r=1}^{\infty }\left\{\frac{1}{2^{N}}\operatorname{tr}\left(\left( Q_{N}^{(g)}\right)^{r}\right)\right\}\frac{u^{r}}{r}                       \\
         & =\sum _{r=1}^{\infty }\sum _{k=0}^{N-1}\binom{N-1}{k}\left(\frac{1}{2}\right)^{N-1-k}\left(\frac{1}{2}\right)^{k} t^{rk}\frac{u^{r}}{r}              \\
         & =\sum _{k=0}^{N-1}\binom{N-1}{k}\left(\frac{1}{2}\right)^{N-1-k}\left(\frac{1}{2}\right)^{k}\sum _{r=1}^{\infty }\frac{\left( t^{k} u\right)^{r}}{r} \\
         & =-\sum _{k=0}^{N-1}\log\left( 1-t^{k} u\right) P( X_{N-1} =k)                                                                                        \\
         & =-E\left[\log\left( 1-t^{X_{N-1}} u\right)\right].
    \end{align*}
\end{proof}
\begin{remark}
    The previous study \cite{IPS} shows a similar result for $QCA$. The model is defined with the local operator
    \[
        Q_{QCA,1}^{(l)} (\xi ,\xi )=\left[\begin{array}{ c c c c }
                \cos \xi & \cdot    & -\sin \xi & \cdot     \\
                \cdot    & \cos \xi & \cdot     & -\sin \xi \\
                \sin \xi & \cdot    & \cos \xi  & \cdot     \\
                \cdot    & \sin \xi & \cdot     & \cos \xi
            \end{array}\right] .
    \]
    Then,
    \[
        \begin{aligned}
             & \log\left\{\overline{\zeta }\left(Q_{QCA,1}^{(l)} (\xi ,\xi ),\mathbb{P}_{N} ,u\right)\right\} =-E[\log\{1-\exp(iS_{N-1} \xi ) \cdot u\}],
        \end{aligned}
    \]
    where
    \[
        P(S_{n} =2k-n) =\frac{1}{2^{n}}\binom{n}{k}.
    \]
\end{remark}

\section{Summary}
\label{sec04}
In this paper, we focused on the analysis of the zeta function of IPS including the DK model as a special case. After giving the definition of the model and zeta function in Section \ref{sec02}, we investigated the eigenvalues and traces of the $global$ operator of the model in Section \ref{sec03}. Lemma \ref{lem} and Corollary \ref{cor} showed the important property of the $global$ operator, which is the key result of the subsequent discussions. Propositions \ref{prop1} and \ref{prop2} gave the two different solutions for the trace of $global$ operator. Then, we proved the relation of the eigenvalues in Theorem \ref{theo2}. For the case $t=a_{10}^{10} -a_{10}^{00} =a_{11}^{11} -a_{11}^{01}$, eigenvalues and the zeta function are explicitly derived in Theorem \ref{theo3}.
\section*{Acknowledgment}
This work was supported by JSPS KAKENHI Grant Number JP22J20050.
\bibliographystyle{style}
\bibliography{reference}
\end{document}